\documentclass[a4paper,USenglish]{lipics-v2021}

\usepackage{amsmath}
\usepackage{amsfonts}
\usepackage{amssymb}
\usepackage{amsthm}
\usepackage{mleftright}
\usepackage{xspace}
\usepackage{xcolor}

\nolinenumbers

\newcommand{\localKTOne}{$\text{KT}_1$ LOCAL}

\newcommand{\KTZero}{$\text{KT}_0$}
\newcommand{\KTOne}{$\text{KT}_1$}
\newcommand{\cC}{ \mathcal{C}}

\def\ShowComment{True}
\ifdefined\ShowComment

\def\gopal#1{{\color{red}\underline{\textsf{Gopal:}}} {\color{blue} \emph{#1}}}
\def\fabien#1{{\color{orange}\underline{\textsf{Fabien:}}} {\color{blue} \emph{#1}}}

\else
\def\john#1{}
\def\gopal#1{}
\def\fabien#1{}
\fi

\author{Fabien Dufoulon}{Lancaster University, UK}{f.dufoulon@lancaster.ac.uk}{https://orcid.org/0000-0003-2977-4109}{Part of this work was done while visiting the University of Padova, partially supported by the ``National Group for Scientific Computing'' (GNCS-INdAM).}
\author{Gopal Pandurangan}{University of Houston, USA}{gopal@cs.uh.edu}{https://orcid.org/0000-0001-5833-6592}{Supported in part by ARO Grant W911NF-231-0191 and NSF grant CCF-2402837.}
\author{Peter Robinson}{Augusta University, USA}{perobinson@augusta.edu}{https://orcid.org/0000-0002-7442-7002}{Supported in part by National Science Foundation (NSF) grant CCF-2402836.}
\author{Michele Scquizzato}{University of Padova, Italy}{scquizza@math.unipd.it}{https://orcid.org/0000-0002-9108-2448}{Supported in part by the Italian National Center for HPC, Big Data, and Quantum Computing.}

\authorrunning{F. Dufoulon, G. Pandurangan, P. Robinson, and M. Scquizzato}

\Copyright{Fabien Dufoulon, Gopal Pandurangan, Peter Robinson, and Michele Scquizzato}

\title{The Singular Optimality of Distributed Computation in LOCAL}

\ccsdesc[500]{Theory of computation~Distributed algorithms}
\ccsdesc[500]{Mathematics of computing~Probabilistic algorithms}
\ccsdesc[300]{Mathematics of computing~Discrete mathematics}

\keywords{Distributed algorithms, round and message complexity, BFS tree construction, leader election} 

\hideLIPIcs
\EventEditors{Silvia Bonomi, Letterio Galletta, Etienne Rivi\`{e}re, and Valerio Schiavoni}
\EventNoEds{4}
\EventLongTitle{28th International Conference on Principles of Distributed Systems (OPODIS 2024)}
\EventShortTitle{OPODIS 2024}
\EventAcronym{OPODIS}
\EventYear{2024}
\EventDate{December 11--13, 2024}
\EventLocation{Lucca, Italy}
\EventLogo{}
\SeriesVolume{324}
\ArticleNo{7}

\begin{document}
\date{}
\maketitle

\begin{abstract}
It has been shown that one can design distributed algorithms that are (nearly) \emph{singularly optimal}, meaning they \emph{simultaneously} achieve optimal time and message complexity (within polylogarithmic factors), for several fundamental {\em global} problems such as broadcast, leader election, and spanning tree construction, under the \KTZero{} assumption. With this assumption, nodes have initial knowledge only of themselves, not their neighbors. In this case the time and message lower bounds are $\Omega(D)$ and $\Omega(m)$, respectively, where $D$ is the diameter of the network and $m$ is the number of edges, and there exist (even) deterministic algorithms that simultaneously match these bounds.

On the other hand, under the \KTOne{} assumption, whereby each node has initial knowledge of itself and the  {\em identifiers} of its neighbors, the situation is not clear. 
For the \KTOne{} CONGEST model (where messages are of small size), King, Kutten, and Thorup (KKT) showed that one can solve several fundamental global problems (with the notable exception of BFS tree construction) such as broadcast, leader election, and spanning tree construction with $\tilde{O}(n)$ message complexity ($n$ is the network size), which can be significantly smaller than $m$. Randomization is crucial in obtaining this result. While the message complexity of the KKT result is near-optimal, its time complexity is $\tilde{O}(n)$ rounds, which is far from the standard lower bound of $\Omega(D)$.  An important open question is whether one can achieve singular optimality for the above problems in the \KTOne{} CONGEST model, i.e., whether there exists an algorithm running in  $\tilde{O}(D)$ rounds and  $\tilde{O}(n)$ messages. Another important and related question is whether the fundamental BFS tree construction can be solved with $\tilde{O}(n)$ messages (regardless of the number of rounds as long as it is polynomial in $n$) in \KTOne{}. 

In this paper, we show that in the \KTOne{} LOCAL model (where message sizes are not restricted),  singular optimality is achievable. Our main result is that {\em all} global problems, including BFS tree construction, can be solved in $\tilde{O}(D)$ rounds and  $\tilde{O}(n)$ messages, where both bounds are optimal up to polylogarithmic factors. Moreover, we show that this can be achieved {\em deterministically}.
\end{abstract}

\section{Introduction and Overview}
\label{sec:intro}

The efficiency of distributed algorithms is traditionally measured by their time and message complexities. The time complexity is the number of rounds in the algorithm, whereas the message complexity is the total amount of messages sent during its execution. 
Distributed algorithms that separately optimize for either the number of rounds or the total amount of messages have been studied extensively; more recently, researchers have designed several algorithms that are \emph{simultaneously} (near) optimal with respect to both measures: such algorithms are called \emph{singularly optimal}~\cite{PanduranganRS17}, and are the focus of this paper. 

\medskip

\noindent\textbf{Singular Optimality in \bf \KTZero{}.} 
 In the {\bf \KTZero{}} model (i.e., {\bf K}nowledge {\bf T}ill radius {\bf 0}), also called the {\em clean network model}~\cite{peleg}, where nodes have initial local knowledge of only themselves (and not of their neighbors),  it has been established that one can obtain (near) singularly optimal distributed algorithms, i.e., algorithms that have simultaneously optimal time and message complexity up to polylogarithmic factors, for many fundamental {\em global} problems\footnote{These problems require traversing the entire network to compute their solution and hence take $\Omega(D)$ rounds, where $D$ is
 the network diameter.} such as leader election, broadcast, Spanning Tree (ST), Minimum Spanning Tree (MST), minimum cut, and approximate shortest paths (under some conditions) \cite{jacm15,PanduranganRS20,Elkin20,HaeuplerHW18}.  For problems such as leader election, broadcast, and ST, it has been shown that one can design a singularly optimal algorithm in the \KTZero{} model that takes $\tilde{O}(m)$ messages ($m$ is the number of edges of the network, and $\tilde{O}(\cdot)$ suppresses logarithmic factors) and $O(D)$ rounds ($D$ is the network diameter); both are tight (up to a $\text{polylog}(n)$ factor where $n$ is the number of nodes) due to matching lower bounds that hold even for Monte Carlo randomized algorithms~\cite{jacm15}. In addition,  MST also admits a  singularly optimal algorithm
 \cite{PanduranganRS20,Elkin20} with message complexity
 $\tilde{O}(m)$ and round complexity $\tilde{O}(D+\sqrt{n})$  (both bounds are tight up to polylogarithmic factors). The singular optimality of MST in the \KTZero{} model also implies the singular optimality of many other problems, such as approximate minimum cut and graph verification problems \cite{stoc11}.

\medskip

\noindent\textbf{\KTOne{} Model and the KKT Result.}
On the other hand, in the {\bf \KTOne{}} model (i.e., {\bf K}nowledge {\bf T}ill radius {\bf 1}), in which each node has initial knowledge of itself and the  {\em identifiers}\footnote{We stress that only knowledge of the identifiers of the neighbors is assumed, not other information such as the degree of the neighbors.} of its neighbors, singular optimality of all the above fundamental problems is wide open. The \KTOne{} model arises naturally in many settings, e.g., in networks where nodes know the identifiers
 of their neighbors (as well as other nodes), e.g., on the Internet, where a node knows the IP addresses of other nodes.
 Similarly, in models such as the Congested Clique \cite{podc15} and the $k$-machine model \cite{soda15,journal}, it is natural to assume that each processor knows
 the identifiers of all other processors.
 For the \KTOne{} CONGEST model (where messages are of small size, typically $O(\log{n})$), King, Kutten, and Thorup (henceforth, KKT)~\cite{KingKT15} showed that one can solve several fundamental global problems (with the notable exception of BFS tree construction) such as broadcast, leader election, and spanning tree construction in message complexity of $\tilde{O}(n)$ ($n$ is the network size) which can be significantly smaller than $m$. (This is in contrast to the \KTZero{} model where $\Omega(m)$ is a lower bound
 of the message complexity.) Randomization is crucial in obtaining this result, and the algorithm is {\em not} comparison-based.\footnote{Comparison-based algorithms can operate on identifiers only by comparing them, i.e., given two IDs, one can only determine whether one is less, greater, equal to the other.}
We note that Awerbuch, Goldreich, Peleg, and Vainish~\cite{vainish} show that $\Omega(m)$ is a message lower bound for MST even in the \KTOne{} CONGEST model, for all deterministic algorithms (even non-comparison based) and all randomized (even Monte Carlo) {\em comparison-based} algorithms. The KKT result breaks the $\Omega(m)$ message barrier in \KTOne{} CONGEST by using a randomized {\em non-comparison-based} technique that uses node IDs as inputs to hash functions.
 
While the message complexity of the KKT result is near-optimal (as $\Omega(n)$ is a lower bound on the message complexity even in \KTOne{} --- see, e.g., \cite{PaiPP021}), their {\em round complexity} is $\tilde{O}(n)$, which is far from the standard lower bound of $\Omega(D)$ for the leader election, broadcast, and ST problems.  An important open question is whether one can achieve singular optimality for the above problems in the \KTOne{} CONGEST model, i.e., whether there exists an algorithm running in  $\tilde{O}(D)$ rounds and  $\tilde{O}(n)$ messages. Another important and related question is whether the fundamental BFS tree construction can be solved in $\tilde{O}(n)$ messages  (regardless of the number of rounds as long as it is a polynomial in $n$) in \KTOne{}. 

Similarly, for the MST problem, the KKT algorithm also showed that one can achieve $\tilde{O}(n)$ message complexity, but it is {\em not} time-optimal --- it can take significantly more than $\tilde\Theta(D+\sqrt{n})$ rounds, which is a time lower bound for MST that applies even for Monte Carlo randomized algorithms~\cite{stoc11}. In subsequent work, Mashreghi and King~\cite{MashreghiK17} presented a trade-off between messages and time for MST: a Monte Carlo algorithm that uses $\tilde{O}(\frac{n^{1+\epsilon}}{\epsilon}\log\log n)$ messages and runs in $O(n/\epsilon)$ time for any $1 > \epsilon \geq \log \log n/\log n$. This algorithm also takes  $\Omega(n)$ time.
Hence another natural and fundamental open question is whether we can design an MST algorithm in the \KTOne{} CONGEST model that is singularly optimal, i.e., runs in $\tilde{\Theta}(D+\sqrt{n})$ rounds, while using $\tilde{\Theta}(n)$ messages.
We note that a singularly optimal algorithm for such a problem is far from certain.
Robinson~\cite{RobinsonSODA21} showed a message complexity lower bound for constructing a $(2k-1)$-spanner in the \KTOne{} CONGEST model, which implies that obtaining singularly optimal algorithms for graph spanners of $O(n^{1+1/k+\epsilon})$ edges, for any constant $\epsilon>0$, is impossible in polynomial time.
However, for many fundamental graph problems, including broadcast, ST, leader election, and MST, the quest for singularly optimal algorithms is still unresolved.

\medskip

\noindent\textbf{Time-Message Tradeoffs in \KTOne{}.}
 Gmyr and Pandurangan~\cite{GmyrPanduranganDISC18} presented several results that show that tradeoffs between time and messages in the \KTOne{} CONGEST model are possible for various fundamental problems. (See also  the related results of \cite{ghaffarikuhndisc18}.)
 The time-message tradeoff results are based on a uniform and general approach which involves constructing
 a {\em sparsified spanning subgraph} of the original graph --- called a {\em danner} (i.e., ``diameter-preserving spanner'')
  --- that trades off the {\em number of  edges} with the {\em diameter} of the sparsifier.
 A key ingredient of this approach is a distributed randomized algorithm that, given a graph $G$ and any $\delta \in [0,1]$, with a high probability
constructs a danner that has diameter $\tilde{O}(D + n^{1-\delta})$ and $\tilde{O}(\min \{m,n^{1+\delta} \})$ edges in $\tilde{O}(n^{1-\delta})$ rounds while using $\tilde{O}(\min\{m,n^{1+\delta}\})$ messages, where $n$, $m$, and $D$ are the number of nodes, edges, and the diameter of the input graph $G$, respectively.
Using a danner, they show that the leader election, broadcast, and ST problems can be solved in $\tilde{O}(D + n^{1-\delta})$ rounds using $\tilde{O}(\min\{m,n^{1+\delta}\})$ messages for any $\delta \in [0,1]$. Another important consequence of the danner construction is that the MST and connectivity problems can be solved in $\tilde{O}(D+n^{1-\delta})$ rounds using $\tilde{O}(\min\{m,n^{1+\delta}\})$ messages for any $\delta \in [0,0.5]$. In addition to getting any desired tradeoff (by plugging in an appropriate $\delta$), one can get a time-optimal algorithm by choosing $\delta = 0.5$, which results in a distributed MST algorithm that runs in $\tilde{O}(D+\sqrt{n})$ rounds and uses $\tilde{O}(\min\{m,n^{3/2}\})$ messages. (Note that this result is time-optimal but not message-optimal.)
While these results improve over prior bounds in the \KTOne{} CONGEST model~\cite{KingKT15,MashreghiK17,vainish}, they do not answer the question of whether singularly optimal algorithms exist for fundamental global problems in the \KTOne{} CONGEST model, nor do they show that the tradeoff bounds are tight.

\medskip

\noindent\textbf{Singularly Optimality of Local Problems.} The work of Bitton et al.~\cite{BittonEIK19}
is the most similar in spirit to this work and raises and answers the same questions but for {\em local} problems.
By local problems, we mean problems with $t$-round algorithms for small $t$ (say, polylogarithmic in $n$) in the LOCAL model. In particular, they ask the following question: Given a LOCAL algorithm that runs in $t$ rounds, is it possible to design an algorithm that takes $O(t)$ rounds while sending only $O(n^{1+\epsilon})$ messages for an arbitrarily small constant $\epsilon > 0$? They point out that this question can be resolved if one can construct an $\alpha$-spanner of $G$ with $O(1)$ stretch and $O(n^{1+\epsilon})$ edges in $O(1)$ rounds sending $O(n^{1+\epsilon})$ messages.
They also point out that, despite the extensive literature on distributed spanner construction algorithms, it is not known whether such a LOCAL spanner construction algorithm exists.

They then present message-reduction schemes for LOCAL algorithms under the \KTOne{} assumption\footnote{Actually, their results also hold in a somewhat weaker model where edges have unique IDs which are known to both endpoints.}
that preserve their asymptotic time complexity. They point out that the gossip-based algorithms of \cite{censor2012global, Haeupler15}  imply that a $t$-round algorithm in the LOCAL model can be transformed into an algorithm in the LOCAL model  that runs in $O(t\log n+ \log^2 n)$ rounds and $O(nt\log n+ n\log^2 n)$ messages. 
Then they show how to transform a $t$-round LOCAL algorithm into an $O(t)$-round algorithm that uses $O(tn^{1+\epsilon})$ messages for an arbitrarily small constant $\epsilon > 0$. We refer to \cite{BittonEIK19}
for more details.

An  implication of the work of \cite{censor2012global, Haeupler15} is that {\em local} problems admit (nearly) singularly optimal algorithms.
For example, fundamental local problems such as MIS (Maximal Independent Set), coloring, and maximal matching 
admit $O(\log n)$-round algorithms, and hence one can design algorithms for these problems that run
in $O(\log^2 n)$ rounds and $O(n\log^2 n)$ messages in \KTOne{} LOCAL. 
It is important to point out that the above results {\em do not} help in designing singularly optimal algorithms for
{\em global} problems, which require $\Omega(D)$ rounds. This is the focus of this paper. 

\medskip

\noindent\textbf{Open Questions.}
{\em The motivating question underlying this work is understanding the status of various fundamental {\em global} problems in the \KTOne{} model --- whether they are singularly optimal or exhibit trade-offs (and, if so, to quantify the trade-offs).}
In particular,  an important open question is whether one can design a randomized  (non-comparison-based)  algorithm for broadcast or for leader election that takes $\tilde{O}(D)$ time and $\tilde{O}(n)$ messages in the \KTOne{} CONGEST model. 
Also, King et al.~\cite{KingKT15} ask whether it is possible to construct an ST in $o(n)$ rounds with $o(m)$ messages. As mentioned earlier, algorithms that are message-optimal, i.e., taking $\tilde{O}(n)$
messages, take $O(n)$ rounds. The situation is also wide open from a lower bound point of view: no lower bound is known on one measure conditional on the other measure being optimal. In particular, even $o(n)$-rounds message-optimal algorithms are not known.

Another important and related question is whether the fundamental BFS  tree construction 
can be solved in $\tilde{O}(n)$ messages  (in polynomial number of rounds) in \KTOne{}.   We note that the polynomial rounds restriction is necessary since the KKT result has an important implication for the \KTOne{} CONGEST model in general. Using this result, one can show that \textit{any} problem can be solved in the \KTOne{} CONGEST model using $\tilde{O}(n)$ messages, provided \textit{exponentially} many rounds are allowed~\cite{RobinsonSODA21}.
This algorithm uses the KKT result to first construct a spanning tree and uses the ``time encoding'' trick to upcast the entire graph topology up the tree.
In the ``time encoding'' trick, clock ticks are used to encode information, and a node can convey a lot of information by being silent for many clock ticks and then sending a bit at an appropriately chosen clock tick.

The main question that we address in this paper is whether singularly optimal algorithms for the above-mentioned fundamental global problems are possible in the \KTOne{} LOCAL model (where the size of a message is not restricted).
Surprisingly, the answer is yes! Note that this is not obvious. In particular, the KKT algorithm to construct a spanning tree (even) in the LOCAL model takes  $\tilde{O}(n)$ rounds in the worst case. On the other hand, one can construct a BFS tree and aggregate the entire topology in $O(D)$ rounds which is time optimal, but it is not known how to construct a BFS tree in $\tilde{O}(n)$ messages in \KTOne{} LOCAL.

\subsection{Distributed Computing Model}

As is standard, the communication network is modeled as an undirected graph \(G=(V, E)\), $n = |V|$, $m = |E|$.
Nodes in the graph can thus be viewed as processors or machines, and each node has a unique \texttt{ID} drawn from a space whose size is polynomial in $n$. The undirected edges model communication links.  We assume the  \emph{synchronous} communication model, where both computation and communication proceed in lockstep, i.e., in discrete time steps called \textit{rounds}.  In the CONGEST model we allow
only small message sizes (typically logarithmic in $n$, the number of nodes) to be sent per edge per round, whereas in the LOCAL model the message size is unrestricted.
In each round of the synchronous model, each node (i) receives all messages sent to it in the previous round, (ii) performs arbitrary local computation based on the information it has, and (iii) sends a message to each of its neighbors in the graph. 

For an algorithm $\mathcal{A}$ in the synchronous model (whether LOCAL or CONGEST), its \textit{round complexity} is the number of rounds it takes to finish and produce output, and its \textit{message complexity} is the total number of messages sent by all nodes throughout the algorithm. 

The CONGEST and LOCAL models come in two standard versions, depending on the type of initial knowledge available to nodes.
In the \KTZero{} (\textit{Knowledge Till radius 0})  model, also called the {\em clean network model}~\cite{peleg}, nodes have initial knowledge of only themselves and do not know anything about their neighbors (e.g., \texttt{ID}s of neighbors).
In the \KTOne{}  model,  each node has initial knowledge of itself \textit{and} the \texttt{ID}s of its neighbors. 
So in the \KTOne{}  model, a little knowledge about neighbors comes for free.
 Both  \KTZero{} and \KTOne{} have been used extensively in the distributed computing literature for several decades.
From a round complexity point of view, it is unnecessary to distinguish between the \KTZero{} and  \KTOne{} versions because it takes just one round for nodes to gather \texttt{IDs} of all neighbors. However, this distinction turns out to be quite critical for message complexity as discussed earlier.

\subsection{Our Results and Techniques}\label{sec:results}

\noindent\textbf{Our Main Result.}
 Our main result is that singular optimality is achievable in the \KTOne{} LOCAL model for global problems. This partially answers the central motivating question raised above and shows that unrestricted message sizes allow one to solve various global problems in a singularly optimal fashion in \KTOne{}. 
 Specifically, we show that {\em all} global problems, {\em including} the key BFS tree construction problem, can be solved in $\tilde{O}(D)$ rounds and  $\tilde{O}(n)$ messages, where both bounds are optimal up to polylogarithmic factors. Furthermore, we show this can be achieved {\em deterministically}.
We also note that, as mentioned earlier, due to a result of~\cite{vainish}, any algorithm (even randomized) that breaks the $\Omega(m)$ message lower bound has to be {\em non-comparison} based  in \KTOne{} CONGEST. In contrast, our algorithms are {\em comparison-based} thus showing that
the $\Omega(m)$ message lower bound of \cite{vainish} does not apply to \KTOne{} LOCAL.
 
 Although our results (detailed below) do not answer the question of whether singular optimality is possible in 
 \KTOne{} CONGEST, they give several insights into the complexity of solving problems under the \KTOne{} assumption. 

\medskip

\noindent\textbf{Singularly Optimal Algorithms.}
 We present two singularly optimal algorithms, one randomized and one deterministic, for solving the BFS tree construction problem, our main technical contribution. Using the BFS tree construction, we show how to solve several other problems, including leader election. The two algorithms have similar complexities but follow two different approaches, illustrating the power of \KTOne{} LOCAL. 

\medskip

\noindent\textbf{A Local Approach.}
 In the first approach (cf. Section \ref{sec:singularlyOptimalBFS}),  we build a BFS tree in the usual ``level-by-level'' fashion starting from the root node. This clearly takes $O(D)$ rounds. However, if done naively, this BFS exploration from level $i$ to level $i+1$ (root is level 0)  takes a number of messages proportional to the number of {\em edges} between the two levels as well as the edges between nodes in level $i$. This leads to $O(m)$ message complexity overall, where $m$ is the total number of edges.  
 
 To obtain $\tilde{O}(n)$ message complexity instead, a crucial idea is to use a {\em sparse neighborhood cover} to ensure that the number of messages needed to {\em grow the BFS tree from level $i$ to level $i+1$}  is essentially (up to the $O(\log n)$ factor) proportional to the number of {\em  BFS tree edges} between the levels. Importantly,  we show that sparse neighborhood cover construction 
 (which is done as a preprocessing step) itself can be done singularly optimally in $O(\log^3 n)$ rounds and $O(n\log^3 n)$
 messages in \KTOne{} LOCAL (cf. Section \ref{subsec:randomizedNeighborhoodCovers}). 
 We refer to Section \ref{sec:singularlyOptimalBFS} for details regarding how the cover is useful in avoiding sending unnecessary messages through non-BFS tree edges. This approach is randomized since we use the sparse neighborhood cover construction algorithm
 due to Elkin \cite{Elkin06}, which is randomized. 
 We modify this algorithm to use $\tilde{O}(n)$ messages (we note that the message complexity of the original implementation is $\tilde{O}(m)$ \cite{Elkin06}). As noted at the end of Section \ref{sec:singularlyOptimalBFS}, we could have also used a deterministic cover construction scheme. However, this has a significantly higher (in terms of polylog factors) message and time complexities 
 compared to the randomized algorithm and a more efficient deterministic algorithm that uses a global approach (described below).

 The main feature of the local approach algorithm is that BFS tree construction can be done in a fashion where {\em each node
 needs to know information only from its local neighborhood}. Specifically, the cover construction algorithm 
 and the subsequent BFS exploration can be implemented in such a way that each node needs to know the topology information of nodes (including its IDs and state information) of only its $O(\log n)$ neighborhood (cf. Section \ref{sec:singularlyOptimalBFS}).

\medskip

\noindent\textbf{A Global Approach.}
Our second algorithm (cf. Section \ref{sec:deterministicSolution}) for BFS tree (which is also deterministic) uses a different approach.  It first uses a deterministic {\em gossip-based} algorithm due to Haeupler~\cite{Haeupler15} to construct a sparse spanning subgraph $H$ of the original graph $G$.

In gossip, each node, in each round, can send at most {\em one} message (of arbitrary size) to one neighbor. Gossip-based communication is very lightweight (only $O(n)$ messages are sent by all nodes in a round) and has been studied extensively~\cite{KempeDG03}. Here we use a result due to Haeupler which showed that each node can (locally) broadcast a message to {\em all} nodes within a distance $k$ in $O(k\log n + \log^2 n)$ rounds in the \KTOne{} LOCAL model~\cite{Haeupler15}. Let the edges used in the  Haeupler algorithm to send messages constitute the subgraph $H$.
We show that  $H$ has only $O(n\log n)$ edges and is
a $O(\log n)$-spanner, i.e., all distances in $H$ between each of pair of nodes are at most a $O(\log n)$ factor of the (true) distance in $G$ \cite{peleg}. In particular, the  property we need is that the diameter of $H$ is $O(D \log n)$ where $D$ is the diameter of $G$. Once $H$ is built, one can simply do a usual flooding-based BFS tree construction (from the given source node) {\em in $H$} as in the \KTZero{} CONGEST model (e.g., see \cite{peleg}). Since $H$ has $O(n\log n)$ edges and $O(D\log n)$ diameter, these determine the same message and time complexities respectively for the BFS tree construction in $H$. 

Once a BFS tree $T_H$ has been constructed in $H$, we use it to collect the entire topology of $G$ at the root of $T_H$  by convergecasting. In the LOCAL model, this takes $O(D\log n)$ rounds and $O(n\log n)$ messages. 
Since the root has the entire topology, it can construct a BFS tree on $G$ locally, and then broadcast the solution to all nodes in $O(D\log n)$ and $O(n\log n)$ messages using the BFS tree of $H$.

This approach to BFS construction is ``global'' since the root needs the entire graph's topology information.

\medskip

\noindent\textbf{Leader Election.}
Both approaches enable one to solve leader election efficiently in \KTOne{} LOCAL as well. If we allow randomization (cf. Section \ref{sec:singularlyOptimalBFS}), we can start with $\Theta(\log n)$
candidates (this can be achieved by each node becoming a candidate with probability $\Theta(\log n/n)$).
Then each of the candidates can run the BFS tree construction algorithm in parallel to elect a leader among the
$\Theta(\log n)$ candidates (e.g., the one with the highest ID). 

In the deterministic setting (cf. Section \ref{sec:deterministicSolution}), we simply run a deterministic leader election algorithm of Kutten et al.~\cite{jacm15} on the $O(\log n)$-spanner $H$ of $G$ constructed earlier.

\medskip

\noindent\textbf{Global Problems.}
In both approaches, one can solve any global problem (i.e., a problem that requires $\Omega(D)$ rounds) in \KTOne{} LOCAL in a singularly optimal fashion, i.e., in $\tilde{O}(D)$ rounds and $\tilde{O}(n)$ messages, as follows. Once a leader is elected and a BFS tree is built, one can collect the entire topology by convergecasting it to the root in $O(D)$ rounds and using $\tilde{O}(n)$ messages. The root solves the problem locally and broadcasts the solution with the same time and message bounds. 

\medskip
 
\noindent\textbf{Summary of Results.}
To summarize, we present the following results.
\begin{itemize}
\item A distributed randomized Monte Carlo algorithm that constructs a BFS tree (given a designated root node) in time $O(D \log n + \log^3 n)$ and using $O(n \log^3 n)$ messages in the \KTOne{} LOCAL model. 
\item A distributed deterministic algorithm that constructs a BFS tree (given a designated root node) in time $O(D \log n + \log^2 n)$ and using $O(n \log^2 n)$ messages in the \KTOne{} LOCAL model. While we note that the bounds for the deterministic algorithm for BFS are better than the corresponding randomized algorithm, the deterministic algorithm follows the global approach where the root node learns information about the entire topology, unlike the randomized algorithm that follows a local approach where nodes learn information only within their $O(\log n)$-radius neighborhood.
\item A distributed randomized Monte Carlo algorithm that elects a leader in time $O(D \log n + \log^3 n)$ and using $O(n \log^4 n)$ messages in the \KTOne{} LOCAL model. 
\item A distributed deterministic algorithm that elects a leader in time $O(D \log^2 n + \log^2 n)$ and using $O(n \log^2 n)$ messages in the \KTOne{} LOCAL model. 
\item All global problems, including broadcast, MST, shortest paths, minimum cut, etc., can be solved in the above (leader election) randomized and deterministic time and message bounds, respectively. 
\end{itemize}

\subsection{Other Related Work}

In the \KTOne{} model,  the early work
of Awerbuch et al.~\cite{vainish} studied time-message trade-offs for broadcast.
King et al.~\cite{KingKT15} surprisingly showed that the basic $\Omega(m)$ message lower bound that holds in the \KTZero{} model for various problems such as leader election, broadcast, MST, and more~\cite{jacm15} does not hold in the \KTOne{} model by showing a randomized Monte Carlo algorithm to construct an MST (which also holds for leader election and broadcast) in $\tilde{O}(n)$ messages  and in $\tilde{O}(n)$ time. Their algorithm uses the powerful randomized technique of {\em graph sketching} to identify edges going out of a cut efficiently, without probing all the edges in the cut; this crucially helps in reducing
the message complexity. In contrast, our work shows that graph sketching is not necessary to obtain $\tilde{O}(n)$
message complexity in \KTOne{} LOCAL. 

There has also been work in understanding the message complexity of local problems such as MIS, coloring,
and maximal matching in \KTZero{}
and \KTOne{}. In the \KTZero{} CONGEST  model, Pai et al. \cite{disc17} showed that any MIS algorithm requires $\Omega(n^2)$ messages. The same quadratic lower bound  was shown to hold for $(\Delta+1)$-coloring \cite{PaiPP021} and also for maximal matching \cite{itcs24}. These lower bounds are all tight, to within a logarithmic factor, since Luby's algorithm solves these problems using $\tilde{O}(m)$ messages.

In contrast, the message complexity of these problems is largely unknown in the \KTOne{} CONGEST model.
Pai et al.~\cite{PaiPP021} showed an $\Omega(n^2)$ lower bound for  \textit{comparison-based} algorithms for MIS, $(\Delta+1)$-coloring, and maximal matching in this model. 
However, whether one can obtain an $o(m)$-message algorithm  is open for arbitrary algorithms. It is worth noting that a quadratic lower bound for comparison-based algorithms does not constitute strong evidence that the problem may have a quadratic lower bound in general.
For example, it was  shown that $(\Delta+1)$-coloring can also be solved using $\tilde{O}(n^{1.5})$ messages in $\tilde{O}(D + \sqrt{n})$ rounds 
in the \KTOne{} CONGEST model~\cite{PaiPP021}.

The quest for singularly optimal distributed algorithms has also been investigated in the asynchronous setting~\cite{MashreghiK21,KuttenMPP21,DufoulonKMPP22} and with more relaxed assumptions on the nodes' initial knowledge~\cite{JiP24}.

Recently there has been a lot of interest in  designing {\em awake or energy-efficient} algorithms 
for various problems in the LOCAL, CONGEST, and radio network models---see, e.g., \cite{podc2020,BM21,DMP23,
ghaffari-podc2023,FMRT23,CDHHLP18, dani22,trygub,sleeping-mst} and references therein. In these works, the goal is to minimize the number of rounds a node is active (when it could be sending/receiving messages). It might be interesting to study singular optimality in the context of awake or energy-efficient algorithms,
where the goal is to simultaneously minimize the message, round, and awake/energy complexities.

\section{Message-Efficient Auxiliary Primitives}
\label{sec:auxPrimitives}

We first provide several primitives that play a crucial role in our randomized singularly optimal BFS \localKTOne{} algorithm (see Section \ref{sec:singularlyOptimalBFS}). In particular, the singularly optimal sparse neighborhood cover construction in Subsection~\ref{subsec:randomizedNeighborhoodCovers} may be of independent interest.

\subsection{Cluster Communication}
\label{subsec:clusterCommunication}

We give some message-efficient primitives for cluster communication: first within a cluster, then with nodes outside. We assume each cluster has an associated (spanning) cluster tree, whose root and edges are (distributedly) known to the network.

Communication within the cluster can be done using standard \emph{broadcast} and \emph{convergecast}. In the broadcast operation, the root holds a message (unlimited in size, as we are in \localKTOne{}) and sends it to its children in the cluster tree. These children in turn send the same message to their children, until it reaches the leaves of the cluster tree and terminates. The convergecast is essentially the reverse operation, where each node contains some (possibly unique) message and sends it to its parent. That parent waits until it receives one message per child, then concatenates its children's messages (without running into bandwidth constraint as we are in \localKTOne{}) and sends that combined message to its own parent. Once the root node receives messages from all of its children, the operation terminates. The runtime and message complexity are straightforward to bound in the \localKTOne{} model:

\begin{proposition}
\label{prop:broadcastConvergecast}
    The broadcast and convergecast operations, over some cluster $C$ with a cluster tree of depth $d$, respectively take $d$ rounds and use $|C|$ messages in \localKTOne{}.
\end{proposition}

For some cluster, say $C$, to communicate message-efficiently with its immediate neighbors (i.e., nodes outside $C$ but with a neighbor in $C$), it suffices to augment the cluster tree with just enough additional edges to ensure it also spans all of these immediate neighbors -- which we call the \emph{outer boundary nodes}. These edges can be computed as follows. 
First, the cluster convergecasts using the (original) cluster tree; each cluster node's message contains its ID and that of all its neighbors, known due to the \KTOne{} assumption. When the convergecast terminates, the root knows the IDs of all cluster nodes as well as the IDs of both endpoints of all edges incident to any cluster node. In particular, the root knows which edges lead out of the cluster. With that information, the root computes a \emph{minimal outgoing edge set}: an edge set such that each node in the outer boundary of $C$ is the endpoint of exactly one edge in the set. (In some cases, we can even ask for the \emph{lexicographically first minimal outgoing edge set}: for each node in the outer boundary, its incident edge in that set is the lexicographically first edge among its incident edges leading to $C$.) Finally, the root broadcasts that minimal outgoing edge set over the (original) cluster tree. Cluster nodes add any edge in the minimal outgoing edge set to its (distributed representation) of the augmented cluster tree (and the other endpoints of these edges are also informed). After which, it is easy to see that the broadcast and convergecast operations on the augmented cluster tree allow the cluster $C$ to broadcast messages to, or convergecast messages from, its outer boundary. Similarly to the previous primitive, the runtime and message complexity are straightforward to bound in the \localKTOne{} model:

\begin{proposition}
\label{prop:outsideCommunication}
    Let $C$ be some cluster with a cluster tree of depth $d$. Then, computing the augmented tree, as well as broadcast and convergecast over that augmented tree of $C$ (i.e., over $C$ and the outer boundary $B$), takes $O(d)$ rounds and $O(|C|+|B|)$ messages in \localKTOne{}.
\end{proposition}

\subsection{Naive BFS Exploration}

The above-described message-efficient cluster communication primitives give a naive but message-efficient primitive for BFS exploration, rooted in some node $r$, up to some small depth $h \geq 1$, say $h = \tilde{O}(1)$.
Let us describe this primitive, called BFSExploration, in more detail. A cluster (with a corresponding cluster tree) is built around the root $r$ in $h$ phases, one layer per phase. In phase $i \in [1,h]$, the cluster computes an augmented cluster tree and broadcasts a BFS exploration message (containing the ID of the cluster's root) over the augmented tree. Cluster nodes follow the broadcast protocol but otherwise ignore the BFS exploration message, whereas the outer boundary nodes join the BFS tree upon receiving a BFS exploration message---the edge along which they receive that message is the edge leading to their parent in the BFS tree.

\begin{proposition}
\label{prop:BFSExploration}
    Let $C$ be the cluster obtained via BFSExploration for some $h \geq 1$. Then, BFSExploration takes $O(h^2)$ rounds and uses $O(|C| \cdot h)$ messages.
\end{proposition}

\begin{proof}
    For any $i \geq 1$, phase $i$ starts with a cluster tree of depth $i$, computes its augmented tree and broadcasts an exploration message over it. By Proposition \ref{prop:outsideCommunication}, this takes $O(i)$ rounds and $O(|C|)$ messages (as no node ever leaves the cluster). The  statement follows.
\end{proof}

This naive BFS exploration algorithm uses few messages when $h$ is small, say $\tilde{O}(1)$, and thus builds a BFS tree message-efficiently for small diameter graphs. However, in general graphs, both its round complexity and message complexity are too high. In the next subsection we use that naive algorithm to obtain sparse neighborhood covers, and these covers in turn allow us to obtain (near) singularly-optimal BFS algorithms in the \localKTOne{} model.

\subsection{Randomized Sparse Neighborhood Cover Computation}
\label{subsec:randomizedNeighborhoodCovers}

Sparse neighborhood covers, defined next, were introduced in~\cite{AP90} and have found various applications in distributed computing such as in routing and in the design of efficient synchronizers. Here, they play a crucial role within our singularly-optimal BFS algorithm (see Section \ref{sec:singularlyOptimalBFS}). Below, we define sparse neighborhood covers.

\begin{definition}\label{def:cover}
A \emph{sparse $(\kappa,W)$-neighborhood cover} of a graph is a collection $\cC$ of trees,
each called a {\em cluster}, with the following properties.
\begin{itemize}
\item \emph{(Depth property)} For each tree $\tau \in \cC$, depth$(\tau) = O(W \cdot \kappa)$.
\item \emph{(Sparsity property)} Each vertex $v$ of the graph appears in $\tilde O(\kappa \cdot n^{1/\kappa})$ different trees $\tau \in \cC$.
\item \emph{(Neighborhood property)} For each vertex $v$  of the graph there exists a tree $\tau \in \cC$ that contains the entire $W$-neighborhood of vertex $v$.
\end{itemize}
\end{definition}

Next, we give a randomized, message-efficient algorithm, called CoverConstruction, for computing a sparse neighborhood cover in \localKTOne{} --- i.e., $\tilde{O}(n)$ message complexity. We do so by a simple modification of the distributed (randomized) cover construction due to Elkin~\cite{Elkin06}. Note that this modification is necessary, since the cover construction given in~\cite{Elkin06} uses $\Omega(m)$ messages, even in \localKTOne{} (which is considered message-inefficient here).

Let us briefly describe the distributed (randomized) cover construction of Elkin~\cite{Elkin06}. The algorithm runs in $\kappa$ phases, each of $O(\kappa \cdot n^{1/\kappa} \log n \cdot \kappa W)$ rounds. Initially, nodes are uncovered. Then, in each phase $i \in [1,\kappa]$, a well-chosen number of uncovered vertices initiate a BFS exploration of depth $2((\kappa-i)+1) W$, each such exploration forming a cluster of the cover. (The property of the BFS exploration primitive ensures the depth property of the cover, while the judicious choice of the number of sources per phase ensures the sparsity property---see Lemma~\ref{lem:smallMaxDegreeOfCover} next.) Nodes that join such a cover, within at most $2(\kappa-i) W$ hops of the root, becomes covered; indeed, its $W$-hop neighborhood is contained in that cluster, and hence in the cover, thus satisfying its neighborhood property.

CoverConstruction modifies the above algorithm as follows: the (message-inefficient) naive BFS exploration primitive used in that algorithm is replaced with our (message-efficient) BFSExploration primitive. It follows that most properties of distributed (randomized) cover construction due to Elkin~\cite{Elkin06} directly translate over to the modified version. In particular, the correctness of CoverConstruction follows directly. Additionally, it also holds for CoverConstruction that with high probability, the maximum degree of the constructed cover is small (see Lemma \ref{lem:smallMaxDegreeOfCover}); in other words, each vertex appears in a small number of clusters of the cover, in fact in $O(\kappa \cdot n^{1/\kappa} \log n)$ of them. This last property will be the key to bounding the message complexity of our CoverConstruction algorithm, along with Proposition \ref{prop:BFSExploration}. 

\begin{lemma}[Corollary A.6 in \cite{Elkin06}]
\label{lem:smallMaxDegreeOfCover}
    The sparse neighborhood cover $\cC$ computed by CoverConstruction satisfies the following: with high probability, for every vertex $v \in V$,
    $|\{C \in \cC \mid v \in C\}| = O(\kappa \cdot n^{1/\kappa} \log n)$.
\end{lemma}

\begin{lemma}
\label{lem:neighborhoodCover}
    There exists a distributed randomized Monte Carlo algorithm that constructs (with high probability) a $(\kappa, W)$-sparse neighborhood cover in time $O(\kappa^3 W^2)$ and using $O(n \cdot \kappa \cdot n^{1/\kappa} \log n \cdot \kappa W)$ messages in the \localKTOne{} model. 
\end{lemma}

\begin{proof}
    The correctness (with high probability) follows from that of Elkin's algorithm. Next, we bound the round complexity. Each phase runs BFS explorations up to $2 \kappa W$ hops. By Proposition \ref{prop:BFSExploration}, these BFS explorations take $O((\kappa W)^2)$ rounds. This amounts, over $\kappa$ phases, to $O(\kappa^3 W^2)$ rounds.
    As for the message complexity, it suffices to note that messages are only sent during the BFS explorations, and those are of depth at most $2 \kappa W$. By Proposition \ref{prop:BFSExploration}, the BFS exploration that computes some cluster $C$ uses $O(|C| \kappa W)$ messages. By Lemma \ref{lem:smallMaxDegreeOfCover}, each node is contained in at most $O(\kappa \cdot n^{1/\kappa} \log n)$ clusters of the cover (over all phases). Hence, if $\cC$ is the computed cover, then $\sum_{C \in \cC} |C| = O(n \cdot \kappa \cdot n^{1/\kappa} \log n)$ and thus the message complexity of CoverConstruction is $O(n \cdot \kappa \cdot n^{1/\kappa} \log n \cdot \kappa W)$.
\end{proof}

\section{Singularly-Optimal BFS using Sparse Neighborhood Covers}
\label{sec:singularlyOptimalBFS}

We provide a BFS tree construction algorithm (called BFSConstruction) that takes $\tilde{O}(D)$ rounds and $\tilde{O}(n)$ messages in \localKTOne{}. This is the first singularly-optimal BFS algorithm in \localKTOne{}, and the main difficulty lies in obtaining $\tilde{O}(n)$ message complexity without blowing up the runtime (say, to linear in $n$ or more). The improved message complexity is obtained using sparse neighborhood covers, which can be computed using CoverConstruction (see Subsection \ref{subsec:randomizedNeighborhoodCovers}) for example. Note that sparse neighborhood covers were previously used in asynchronous BFS tree construction algorithms~\cite{AP90a,APPS92} for synchronization purposes, whereas we instead use them in a synchronous setting to reduce messages (which  is new to  the best of our knowledge).

We now describe our BFSConstruction algorithm. A BFS tree is built around some (initially given) root node $r \in V$, layer by layer. First, nodes build a ($\kappa$, 2)-cover of $G$, for some well-chosen $\kappa = \Theta(\log n)$. After which, the BFS tree is built in (up to) $D$ phases, each of $O(\log n)$ rounds. In each phase, each node neighboring the current BFS tree receives exactly one message and joins the BFS tree; this adds a layer to the BFS tree in a message-efficient manner. We do this by computing a (lexicographically-first) minimal outgoing edge set (see Subsection \ref{subsec:clusterCommunication}) of the current BFS tree (or cluster), and sending BFS exploration messages via this edge set. 

More concretely, in phase $i \geq 1$, frontier nodes (i.e., on layer $i-1$) of the current BFS tree (or cluster) compute the lexicographically-first minimal outgoing edge set $O_i$ in two stages. The first stage takes $2 \kappa = O(\log n)$ rounds, during which all frontier nodes \emph{ping} each cluster of the cover they belongs to: that is, they send a message to their parent, who in turn sends it to their parent, and so on, ensuring the messages is sent all the way up to the root of the (cover) cluster's tree. (Note that this is enough time for the message to reach the root, since each cover cluster tree has depth $2 \kappa$.) The second stage takes $6 \kappa = O(\log n)$ rounds, during which the root of the cover cluster broadcasts a message to all nodes which then starts a convergecast; nodes in the tree convergecast their ID, whether they are in the BFS, and the information about their incident edges (that is, the IDs of the two endpoints) and upon termination of the convergecast, the root broadcasts all of the aggregated information to all nodes in the tree. 
Once these two stages are over, or in other words, once some frontier node $v$ has received the aggregated information from all cover clusters it belongs to, node $v$ knows its 2-hop neighborhood information as well as which node within that neighborhood is in the BFS tree. That information is sufficient for $v$ to know which of its incident edges belongs to the lexicographically-first minimal outgoing edge set $O_i$ of the current BFS tree (or cluster).

\begin{theorem}
    There exists a distributed randomized Monte Carlo algorithm that constructs a BFS tree (given a designated root node) in time $O(D \log n + \log^3 n)$ and using $O(n \log^3 n)$ messages in the \localKTOne{} model. 
\end{theorem}

\begin{proof}
    We start with the correctness. We claim that in each phase $i \geq 1$, the lexicographically-first minimal outgoing edge set $O_i$ is correctly (and distributedly) computed. The correctness of the BFS tree computation follows straightforwardly from that claim. Let us prove the claim now. Each frontier node (say $v$) obtains its 2-hop neighborhood information due to the neighborhood property of the ($\kappa$, 2)-cover: one of the cover $v$ is contained in also contains all nodes in the 2-hop neighborhood of $v$. Then, $v$ pings this cover cluster tree (and others) during phase $i$, and during the second stage, the information of $v$'s 2-hop neighborhood (including which nodes are in the BFS) is gathered and then sent to $v$. (Note that sufficient runtime is given to the two stages by Proposition \ref{prop:broadcastConvergecast}, and although $v$ receives different messages from the different cover cluster trees it belongs to, it can simply take the message containing the most information on $v$'s 2-hop neighborhood.) From this information, $v$ can compute which incident edge is in the lexicographically-first minimal outgoing edge set: for each neighbor $w$ of $v$, such that $w$ is not in the BFS, $v$ computes the lexicographically-first incident edge to $w$ with a BFS endpoint, and whether that edge is incident to $v$. 
    
    The round complexity is straightforward; computing the ($\kappa$, 2)-sparse neighborhood cover takes $O(\log^3 n)$ rounds by Lemma \ref{lem:neighborhoodCover}, and the $O(D)$ phases each take $O(\log n)$ rounds. We now bound the message complexity. First, computing the ($\kappa$, 2)-sparse neighborhood cover uses $O(n \log^3 n)$ messages by Lemma \ref{lem:neighborhoodCover}. As for the BFS tree computation, we separate messages into two types: the BFS exploration messages, and the messages sent over the cover cluster trees (for the pinging, broadcasts and convergecast). The former travel only along minimal outgoing edge sets, which guarantees each non-root node receives such a message once, joins the BFS and never receives such a message again. Thus over all phases, at most $n$ BFS exploration messages are sent. For the latter, we first claim that we can bound the number of convergecasts and broadcasts executed per cover cluster tree over the algorithm by $O(\log n)$. Indeed, each cluster cover tree has diameter $O(\log n)$ due to our choice of $\kappa$, and thus its vertices can be contained in at most $O(\log n)$ layers of the BFS tree. Moreover, messages are sent over a cover cluster tree in a given phase if and only if one of its nodes is a frontier node (due to the pinging behavior) in that phase. As such, messages are sent over a given cover cluster tree during at most $O(\log n)$ phases, and within each such phase at most $O(1)$ convergecasts and broadcasts are executed, thus proving the claim. (In the above counting, we use the fact that the pinging operation sends less messages than a convergecast.) Given the claim, Proposition~\ref{prop:broadcastConvergecast} and since each node belongs to $O(\log^2 n)$ different cover cluster trees by Lemma \ref{lem:smallMaxDegreeOfCover}, the message complexity due to the second type of messages is $O(n \log^3 n)$.
\end{proof}

The above (near) singularly-optimal BFS algorithm can be transformed into a randomized leader election algorithm via an additional $\Theta(\log n)$ factor in the message complexity as follows. 
Assuming nodes have some good estimate (i.e., within a constant factor) of the network size $n$, $\Theta(\log n)$ candidates are picked uniformly at random. Then, these candidates each start an independent execution of the above BFS algorithm, and messages of different executions are differentiated by simply adding to the messages the candidate's ID.

\begin{corollary}
    There exists a distributed randomized Monte Carlo algorithm that elects a leader in time $O(D \log n + \log^3 n)$ and using $O(n \log^4 n)$ messages in the \localKTOne{} model. 
\end{corollary}

Note also that the above (near) singularly-optimal BFS algorithm can be made deterministic, but at the cost of higher logarithmic factors in both time and message complexities. We leave out the details however, because the deterministic solution given in the subsequent section is strictly more efficient in both time and message complexities. Instead, we give a brief sketch. The crucial point is that the randomized sparse neighborhood cover construction given in Section \ref{subsec:randomizedNeighborhoodCovers} is the only randomized component used in the above BFS construction. We can replace it with a message-efficient deterministic sparse neighborhood cover construction to obtain a message-efficient deterministic BFS construction algorithm. More concretely, one can combine the $(O(\log^3 n), W)$-sparse neighborhood cover construction of \cite{GhaffariT23} (that takes as parameter any positive integer $W$) with the simulation of \cite{censor2012global, Haeupler15} (as described in more detail in both Sections \ref{sec:intro} and \ref{sec:deterministicSolution}). Doing so gives a deterministic $(O(\log^3 n), W)$-sparse neighborhood cover construction algorithm that takes $O(W \log^{11} n)$ rounds and $O(n \log^{11} n)$ messages, and such that each cluster tree has depth $O(\log^3 n)$ and each edge appears in at most $O(\log^4 n)$ cluster trees. Moreover, if used in the above BFS construction algorithm, the resulting time and message complexities are respectively $O(D \log^3 n + \log^{11} n)$ and $O(n \log^{11} n)$.

\section{Deterministic Singularly-Optimal BFS via a Global Approach}
\label{sec:deterministicSolution}

We now give an alternative approach by leveraging results from gossip algorithms.
In contrast to the \localKTOne\ model, communication in the gossip model is much more restricted: 
In each round, a node may activate only a single link to some neighbor and the nodes may perform bidirectional exchange of messages (of arbitrary size) over this link. 
Note that it is possible that multiple nodes may activate a link to some node $u$ in the same round. 

Even though a gossip algorithm sends only $O(n)$ messages per round, it was nevertheless shown in \cite{censor2012global} that there is a $O(\log^3n)$-round randomized algorithm that allows all nodes to exchange information with all of their neighbors, which is known as 1-local broadcast.
That is, each node holds some piece of information (a ``rumor'') initially, and the goal of 1-local broadcast is for each node to learn all rumors of its neighbors in the network.
Subsequently, \cite{Haeupler15} gave a deterministic algorithm that achieves an improved upper bound of $O(\log^2 n)$ rounds for 1-local broadcast.
While there are results for solving broadcast efficiently when restricting ourselves to well-connected networks (see \cite{giakkoupis2011tight}), we point out that \cite{censor2012global} and \cite{Haeupler15} work in arbitrary networks and both crucially rely on the \KTOne\ assumption. 

In the algorithm of \cite{Haeupler15}, each node $v$ keeps track of the set $R_v$ of neighbors from which it has not yet received the rumor, and $v$ also keeps track which links it has activated so far (and in which order) in a set $E_v$.
The algorithm proceeds in iterations and, at the start of iteration $i$,  each node $v$ in parallel activates a link to some arbitrary node in $R_v$ (assuming $R_v\ne \emptyset$) and adds this link $(u,v)$ as its {$i$-th link} $l_i$ to $E_v$.
Then, over the next $i$ rounds, $v$ exchanges rumors by activating the links $l_i,\dots,l_1$, i.e., in reverse order in which they were added to $E_v$.
Afterwards, $v$ again exchanges all rumors by activating the links $l_1,\dots,l_i$ once more (this time in chronological order).
The above two batches of link activations are repeated once more in the same iteration to ensure that the information acquired by the nodes is symmetric in the sense that $u$ received $v$'s rumor if and only if $v$ also received $u$'s rumor, thus yielding a total round complexity of $4i$ rounds in iteration $i$. 
\cite{Haeupler15} shows that the algorithm terminates in $O(\log n)$ iterations, resulting in a time complexity of $O(\log^2 n)$ rounds.

We now show that the algorithm of \cite{Haeupler15} implies the existence of a \emph{multiplicative $O(\log n)$-spanner}~\cite{peleg1989graph}, which is a sparse subgraph $H$ of the network $G$ where, for each edge $\{u,v\} \in G\setminus H$ there exists a path of length $O(\log n)$ in $H$.
 
\begin{lemma} \label{lem:haeupler}
There is a deterministic gossip algorithm that terminates in $O(\log^2n)$ rounds and constructs a $O(\log n)$-spanner $H$ with $O(n\log n)$ edges, such that each node knows which of its incident edges are in $H$.
\end{lemma}
\begin{proof}
We define $H$ as the union of the links added to $E_v$, for each node $v$.
Since the algorithm terminates in $O(\log n)$ iterations, it follows that the set of activated links $E_v$ of each node $v$ is of size $O(\log n)$, which implies the claimed upper bound on the number of edges in $H$.

It is shown in \cite{Haeupler15} that a node can exchange rumors with \emph{all} its neighbors by activating the links in $H$ according to the schedule in the last iteration of the algorithm.
As mentioned above, the number of rounds in iteration $i$ is $4i$, and hence the last iteration takes $O(\log n)$ rounds.
In other words, the edges of $H$ ensure the existence of a path of length $O(\log n)$ between $v$ and each one of its neighbors, which shows that $H$ is a spanner with stretch $O(\log n)$. 
\end{proof}

To construct a BFS tree on the network $G$, we first initiate the construction of a BFS tree $T_H$ on $H$ by executing the simple flooding-based algorithm starting at the designated root node $s$. 
According to Lemma~\ref{lem:haeupler}, this takes $O(D\log n)$ rounds and the number of messages sent is $O(|E(H)|) = O(n\log n)$.
Then, we convergecast the entire topological information of each node starting at the leaves in $T_H$ up to the root $s$, thus enabling $s$ to locally compute a BFS tree $T$ on $G$.
Finally, $s$ broadcasts the topological information of $T$ over the edges of $T_H$, which ensures that each node learns its parents and children in $T$.
Since the broadcast and convergecast on $T_H$ both take $O(D\log n)$ rounds and $O(n\log n)$ messages, the complexity of constructing a BFS tree is dominated by the cost of constructing the spanner $H$, which yields the following:

\begin{theorem}
\label{thm:deterministicResult}
  There exists a deterministic algorithm that constructs a BFS tree from a designated source in time $O(D\log n + \log^2n)$ and using $O(n \log^2n)$ messages in the \KTOne{} LOCAL model.
\end{theorem}

If, instead of constructing a BFS tree, our goal is to elect a unique leader among the nodes, we leverage the deterministic \KTZero{} CONGEST leader election algorithm of \cite{jacm15}, which, on general networks requires $O(D\log n)$ rounds and $O(m\log n)$ messages.
When executing this algorithm on the spanner $H$ instead, Lemma~\ref{lem:haeupler} yields the following result:
 
\begin{corollary} \label{cor:deterministicLeader}
Leader election can be solved deterministically in $O(D\log^2n + \log^2n)$ rounds and $O(n\log^2n)$ messages in the \KTOne{} LOCAL model.
\end{corollary}

\section{Conclusion}

In this paper, we fully resolve the complexity of distributed computation in the \KTOne{} LOCAL model, showing
that singularly optimal algorithms can be achieved for global problems. 
A question that we leave open is whether  {\em any} problem that admits a $O(t)$-round algorithm in the LOCAL model, for any $t\leq D$, can be solved in  $\tilde{O}(t)$ rounds and $\tilde{O}(n)$ messages in \KTOne{} LOCAL? In particular, the question is whether we can design singularly optimal algorithms for any problem in \KTOne{} LOCAL. This work answers this question for $t = D$. 

Another interesting question is whether one can improve the complexities of our
\KTOne{} LOCAL algorithms. In particular, can we obtain an $O(D)$-round algorithm for a BFS tree (or leader election and broadcast)  that uses $\tilde{O}(n)$ messages?

Finally, a major open question is whether singular optimality is possible for global problems in
the \KTOne{} CONGEST model. 

\bibliographystyle{plainurl}
\bibliography{ref}

\end{document}